\documentclass[11pt]{article}
\usepackage{amsfonts,amsmath,amssymb,graphicx,tabularx,url,epstopdf,enumerate}

\newcommand{\R}{{\mathbb R}}

\newtheorem{theorem}{\bf Theorem}[section]

\newtheorem{lemma}[theorem]{\bf Lemma}

\newtheorem{conjecture}[theorem]{\bf Conjecture}

\newcommand{\qed}{\hfill $\blacksquare$ \bigskip}

\begin{document}

\title{\bf{Path matrix and path energy of graphs}}

\author{
Aleksandar Ili\' c \\
Facebook Inc, Menlo Park, California, USA \\
e-mail: \tt{aleksandari@gmail.com}
\and
Milan Ba\v si\'c \\
Faculty of Sciences and Mathematics, University of Ni\v{s}, Serbia \\
e-mail: \tt{basic\_milan@yahoo.com}
}

\date{\today}
\maketitle

\begin{abstract}
Given a graph $G$, we associate to it a path matrix $P$ whose $(i, j)$ entry represents the maximum number of vertex disjoint paths between the vertices $i$ and $j$, with zeros on the main diagonal. In this note, we resolve four conjectures from [M. M. Shikare, P. P. Malavadkar, S. C. Patekar, I. Gutman, \emph{On Path Eigenvalues
and Path Energy of Graphs}, MATCH Commun. Math. Comput. Chem. {\bf 79} (2018), 387--398.] on the path energy of graphs and finally present efficient  $O(|E| |V|^3)$ algorithm for computing the path matrix used for verifying computational results.
\end{abstract}

 \vskip 0.1cm

\section{Introduction}
\label{sec:intro}

Let $G$ be a simple graph with vertex set $V (G) = \{v_1, v_2, \ldots, v_n\}$. Define the matrix
$P = (p_{ij})$ of size $n \times n$ such that $p_{ij}$ is equal to the maximum number of vertex disjoint
paths from $v_i$ to $v_j$ for $i \neq j$, and $p_{ij} = 0$ if $i = j$. We say that $P = P(G)$ is the path matrix of the graph $G$ [16]. $P$ is a real and symmetric matrix and therefore has real spectrum
$$
Spec_{\bf P} = [\rho_1, \rho_2, \ldots, \rho_n].
$$

The path spectral radius of a graph is the largest eigenvalue $\rho = \rho(G)$ of the path matrix $P(G)$. Denote by $deg(v)$ the degree of the vertex $v$, and by $\Delta(G)$ the largest vertex degree. The ordinary energy $E(G)$, of a graph $G$ is defined to be the sum of the absolute values of the ordinary eigenvalues of $G$ \cite{Li12}. In analogy, the path energy $PE(G)$ is defined as
$$
PE(G) = \sum_{i = 1}^n |\rho_i|.
$$

Shikare et al. \cite{ShMa18} studied basic properties of the path matrix and its eigenvalues. Here we continue the study by focusing on extremal problems of path energy of graphs, and resolve open problems around the structure of general and unicyclic graphs attaining maximal or minimal values of $PE(G)$.

\section{Extremal values of path energy of graphs}
\label{sec:2}

The authors in \cite{ShMa18} proved the following two simple results.

\begin{theorem}
\label{thm1}
Let $G$ be a connected graph of order $n$. Then
\begin{enumerate}[(i)]
\item $\rho (G) \geq (n - 1)$, with equality if and only if $G$ is a tree of order $n$.
\item $\rho(G) \leq (n - 1)^2$, with equality if and only if $G$ is a complete graph $K_n$.
\end{enumerate}
\end{theorem}

\begin{theorem}
\label{thm2}
Let G be a graph with vertex set $V(G)$. Then for all $u, v \in V$, it holds that
$$p_{uv} (G) \leq \min \{ deg (u), deg(v) \}$$.
\end{theorem}

The following result resolves Conjecture 1 from \cite{ShMa18}.

\begin{theorem}
Let $G$ be a connected graph of order $n$. Then
\begin{enumerate}[(i)]
\item $PE (G) \geq 2(n - 1)$, with equality if and only if $G$ is a tree of order $n$.
\item $PE (G) \leq 2(n - 1)^2$, with equality if and only if $G$ is a complete graph $K_n$.
\end{enumerate}
\end{theorem}

\begin{proof}
The trace of the matrix $P$ is 0 by definition, which is in turn the sum of all its eigenvalues. The first part directly follows from Theorem \ref{thm1},
$$
PE(G) \geq 2 \rho (G) \geq 2 (n - 1).
$$

For the second part we use Cauchy-Schwarz inequality:
$$
PE(G) = \rho(G) + \sum_{i = 2}^{n} \rho_i (G) \leq \rho(G) + \sqrt{(n-1)\left(\sum_{i = 2}^n \rho_i^2\right)}.
$$
Using Theorem \ref{thm2}, the trace of the squared matrix is clearly less than or equal to
$$
Tr(P^2) = \sum_{i = 1}^n \rho_i^2 \leq n(n-1) \cdot \Delta^2 \leq (n-1)^4 n.
$$

Combining the above, we get
$$
PE(G) \leq \rho + \sqrt{(n-1)^4n - \rho^2(n-1)}.
$$

The function $f(x) = x + \sqrt{(n-1)^4n - x^2(n-1)}$ is increasing for $x \leq (n-1)^2$, as the first derivative is non-negative:
$$
f'(x) = 1 - \frac{x(n-1)}{\sqrt{(n-1)^4n - x^2(n-1)}} \geq 0,
$$
which is equivalent with $(n-1)^4n - x^2(n-1) \geq x^2 (n-1)^2$.

The equality holds in both cases if and only if $\rho$ attains minimum or maximum values, which follows from Theorem \ref{thm1}. \qed
\end{proof}

\section{Path energy of unicyclic graphs}
\label{sec:2}

Let $U_{n,k}$ be a unicyclic graph of  order $n$ whose cycle is of  size $k$.
In the following we give  results that resolve Conjectures 2, 3 and 4 from~\cite{ShMa18}.

\begin{conjecture}
\label{con2}
Let $G$ be a unicyclic graph of  order $n$ whose cycle is of  size $k$. Then $PE(G)$ depends only on the parameters $n$ and $k$. For a fixed value of $n$, $PE(G)$ is a monotonically increasing function of $k$.
\end{conjecture}

First we determine the spectrum of the path matrix of unicyclic graphs.
If $k=n$, that is $U_{n,k}\cong C_n$, it is easy to see that $P(U_{n,k})=2J_{n}$,
where $J_n$ is the square matrix of the order $n$ whose all non-diagonal elements are equal to
one, and all diagonal elements are zero. Furthermore, as $J_n$ represents the adjacency matrix of the complete graph $K_n$ we obtain that

$$
Spec_{\bf P} (C_n) = \left(
\begin{array}{cccc}
(-2)^{n-1}, & 2(n-1)^{1}\\
\end{array}\right).
$$

Next we calculate the spectrum of $P(U_{n, k})$ for $k\leq n-1$.

\begin{theorem}
The spectrum of the path matrix of the unicyclic graph $U_{n,k}$, for $k\leq n-1$, is equal to
$$
Spec_{\bf P} (U_{n, k}) = \left(
\begin{array}{cccc}
(-2)^{k-1}, & (-1)^{n-k-1}, & \rho_2^{1}, & \rho_1^{1}\\
\end{array}\right),
$$
where
\begin{eqnarray}
\label{eq:mu2}
\rho_{1,2}=\frac{n+k-3\pm\sqrt{(n+k-3)^2+4(k^2-nk+2n-2)}}{2}.
\end{eqnarray}

\end{theorem}
\begin{proof}
The vertices of $U_{n,k}$ can be labeled so that
$$
P(U_{n,k})=\left[
\begin{array}{cc}
P(C_k) & 1 \\
1 & J_{n-k} \\
\end{array}\right]=
\left[
\begin{array}{cc}
2J_k & 1 \\
1 & J_{n-k} \\
\end{array}\right].
$$
Now, we will determine the roots of the characteristic polynomial
$$
\det(P(U_{n,k})-\lambda I_{n}) =  \left|
\begin{array}{cc}
2J_k-\lambda I_k & 1 \\
1 & J_{n-k} -\lambda I_{n-k} \\
\end{array}\right|,
$$
where $I_s$ is the identity matrix of  order $s$. By subtracting the $s$-th row from the $(s-1)$-th row, for $2\leq s\leq k$, and
by subtracting the $s$-th row from the $(s+1)$-th row, for $k+1\leq s\leq n-1$, we obtain that the $s$-th row is equal to
$$
[\underbrace{0,\ldots,0,\lambda+2,-\lambda-2}_s,0,\ldots,0]
$$
for $2\leq s\leq k$, and the $l-$th row is equal to
$$
[\underbrace{0,\ldots,0,-\lambda-1,\lambda+1}_{l+1},0,\ldots,0]
$$
for $k+1\leq l\leq n-1$. Therefore, we conclude that
$$(\lambda+2)^{k-1}(\lambda+1)^{n-k-1}\mid \det(P(U_{n,k})-\lambda I_{n})$$
and hence it is proved that
$P(U_{n,k})$ has  eigenvalues $-2$ and $-1$ with  multiplicities $k-1$ and $n-k-1$, respectively.
Notice that $n-k-1\geq 0$, as $k+1\leq n$.

Now, we will determine the other two eigenvalues. As the trace of $P(U_{n,k})$ is equal to zero and the trace of the squared matrix $P(U_{n,k})^2$ is equal to the sum of the squares of the entries of $P(U_{n,k})$ we have that
\begin{eqnarray*}
\sum_{i=1}^{n} \rho_i &=& tr(P(U_{n,k}))=0\\
\sum_{i=1}^{n} \rho_i^2 &=&tr (P(U_{n,k})^2)=4(k^2-k)+(n^2-k^2-(n-k)).
\end{eqnarray*}
Furthermore, since we have already concluded that $$(\rho_3,\ldots,\rho_n)=(\underbrace{-2,\ldots,-2}_{k-1},\underbrace{-1,\ldots,-1}_{n-k-1})$$
it holds that
\begin{eqnarray}
\rho_1+\rho_2 &=& 2(k-1)+(n-k-1) \label{eq:sum_eigens}\\
\rho_1^2+\rho_2^2 &=&4(k^2-k)+(n^2-k^2-(n-k))-4(k-1)-(n-k-1).
\end{eqnarray}
By substituting the variable $\rho_2$ from first equation in the second, we get
$$
2\rho_1^2-2(n+k-2)\rho_1-2k^2+2nk-4n+4=0,
$$
which completes the proof.
\qed

\end{proof}

From the above theorem it directly follows that
$$
\rho_{1}=\frac{n+k-3+\sqrt{(n+k-3)^2+4(k^2-nk+2n-2)}}{2}
$$ is the spectral radius of $U_{n,k}$ for $3\leq k\leq n-1$.

By analyzing the above formula we can obtain Propositions 6, 7 and 8 from \cite{ShMa18}.
Indeed, if we denote
$$f(x)=n+x-3+\sqrt{(n+x-3)^2+4(x^2-nx+2n-2)}$$
then it is sufficient to prove that
$f(x)$ is a monotonically increasing function of $x$, for $2\leq x\leq n-1$, to get Proposition 6. The first derivative of $f(x)$ is equal to
$1+\frac{5x-n-3}{\sqrt{D(x)}}$, where $D(x)=(n+x-3)^2+4(x^2-nx+2n-2)$. If $x\geq \frac{n+3}{5}$ then it is clear that $f'(x)\geq 0$.
Now, for $1\leq x<\frac{n+3}{5}$ we prove that $\sqrt{D(x)}>n+3-5x$. After a short calculation, it can be obtained that $D(x)^2>(5x-n-3)^2$ if and only if
$5x^2-2(n+3)x+n+2<0$. Since this quadratic function is convex, it is less than zero in the interval $(x_1,x_2)$, where $x_{1,2}=\frac{n+3\pm\sqrt{(n+3)^2-5(n+2)}}{5}$.
It is easy to check that $x_2\leq 1$ and $x_1\geq \frac{n+3}{5}$ and therefore we conclude $\sqrt{D(x)}>n+3-5x$ and $f'(x)>0$ for $1\leq x<\frac{n+3}{5}$.

From $\frac{f(n-1)}{2}=2(n-1)<(n-1)^2$ (for $n\geq 4$)  follows Proposition 7.
Finally, we can calculate the minimal spectral radius in the class of unicyclic graphs of  order $n$ and it is attained for $k=3$:
$$
\min_{3\leq k\leq n}\rho(U_{n,k}) =\frac{f(3)}{2}=\frac{n+\sqrt{n^2-4n+28}}{2}.
$$

\begin{lemma}
\label{lem:mu2 positive}
The eigenvalue $\rho_2$ of the path matrix of the unicyclic graph $U_{n,k}$, for $k\leq n-1$, is greater than zero if and only if
$n\geq 7$ and $3\leq k\leq n-3$.
\end{lemma}
\begin{proof}
According to (\ref{eq:mu2}) we conclude that $\rho_2>0$ if and only if $k^2-nk+2n-2<0$.
If we denote $k^2-nk+2n-2$ by $g(k)$, then we have that $g(k)<0$ if and only if $k$ belongs to the interval
$(x_1,x_2)$, where $x_{1,2}=\frac{n\pm\sqrt{n^2-8n+8}}{2}\in \R$. Furthermore, $x_{1,2}\in\R$ if and only if
$n^2-8n+8\geq 0$ and this is the case for $(n-4)^2-8\geq 0\Leftrightarrow n\geq 4+\sqrt{8}$. Therefore, $x_{1,2}\in \R$ if and only if $n\geq 7$.
On the other hand, we  notice that
\begin{eqnarray*}
 x_2 &=& \frac{n-4-\sqrt{(n-4)^2-8}}{2}+2>2
\end{eqnarray*}
and
\begin{eqnarray*}
 x_2 &=& \frac{n-6-\sqrt{(n-6)^2+(4n-28)}}{2}+3\leq 3,
\end{eqnarray*}
as $4n-28\geq 0$. From Vieta's formulas we have that $x_1+x_2=n$ and therefore $n-3\leq x_1< n-2$.
Finally, since $k$ is an integer we conclude that $3=\lfloor x_2 \rfloor+1\leq k\leq \lfloor x_1 \rfloor=n-3$.\qed
\end{proof}

Now, we calculate the path energy of a unicyclic graph of the order $n$ and cycle length $k$, such that $k\leq n-1$.
According to the previous lemma the spectrum of $U_{n,k}$ has two positive eigenvalues $\rho_1$ and $\rho_2$
if and only if $n\geq 7$ and $3\leq k\leq n-3$. In that case, the energy of $U_{n,k}$ is equal to $2(\rho_1+\rho_2)$ and
from (\ref{eq:sum_eigens}) we further have that it is equal to $2(n+k-3)$. Now, if $n< 7$ or $n-2\leq k\leq n-1$
then we conclude that $\rho_1$ is the only positive eigenvalue-spectral radius and the path energy is equal to $2\rho_1$.

\begin{theorem}
The path energy of the unicyclic graph $U_{n,k}$ for $3 \leq k\leq n-1$, is equal to
\begin{eqnarray*}
PE (U_{n, k}) &=&\left\{
\begin{array}{rl}
2(n+k-3), & n\geq 7 \mbox{ and } 3\leq k\leq n-3  \\
2\rho(U_{n,k}), & n< 7 \mbox{ or } n-2\leq k\leq n-1 \\
\end{array} \right.
\end{eqnarray*}
where $\rho(U_{n,k})=\rho_1=\frac{n+k-3+\sqrt{(n+k-3)^2+4(k^2-nk+2n-2)}}{2}$ is the spectral radius of $U_{n,k}$.
\end{theorem}

Since it has already been shown that $\rho(U_{n,k})$ is an increasing function of $k$,
and as $2(n+k-3)$ is increasing as well, it remains to prove that
 $$2(n+(n-3)-3)<2\rho(U_{n,n-2}).$$
Indeed, as $2\rho(U_{n,k})=n+k-3+\sqrt{(n+k-3)^2+4g(k)}$, where $g(k)=k^2-nk+2n-2$, we conclude that
$$2\rho(U_{n,k})>n+k-3+\sqrt{(n+k-3)^2+0}=2(n+k-3)>2(n+(n-3)-3),$$
for $n-2\leq k\leq n-1$ (in the proof of Lemma \ref{lem:mu2 positive} we  used the fact that $g(k)>0$ for $n-2\leq k\leq n-1$).
Therefore,  we see that $PE (U_{n, k})$
is an increasing function of $k$ and hence we prove Conjecture \ref{con2} from \cite{ShMa18}.

Conjectures 3 and 4 can be unified and generalized in the following way:
\begin{theorem}
Let $G$ be a unicyclic graph of order $n$. Then
\begin{enumerate}[(i)]
\item $PE (G) \leq 4(n-1)$ with equality if and only if $G \cong C_n$,
\item $PE (G) \geq n+\sqrt{n^2-4n+28}$ with equality if and only if $G \cong U_{n,3}$.
\end{enumerate}
\end{theorem}

\begin{proof}
The above discussion implies that $PE (U_{n, k})$, for $3\leq k\leq n-1$, is maximal for $k=n-1$. Since
$PE (C_{n})=4(n-1)$ it remains to compare the values $PE (U_{n, n-1})$ and $PE (C_{n})$. It can be directly verified that the inequality
$$
PE (U_{n, n-1})=2n-4+\sqrt{(2n-4)^2+4((n-1)^2-n(n-1)+2n-2)}\leq 4(n-1)
$$
holds if and only if  $n\geq 3$. 

Moreover, the minimal path energy in the class of unicyclic graphs of  order $n$ is attained for $k=3$:
$$
\min_{3\leq k\leq n} PE(U_{n,k}) =n+\sqrt{n^2-4n+28}.
$$
\qed
\end{proof}

\section{An efficient algorithm for computing the matrix $P(G)$}
\label{sec:4}



In order to find the maximum number of vertex disjoint paths between two vertices, we can alternatively look at the  problem of finding the maximum number of edge disjoint paths. For all vertices except fixed vertices $x$ and $y$, split vertex $v$ into $v_{in}$ and $v_{out}$ with an edge $v_{in} \rightarrow v_{out}$. If we have an edge $uv$ in the original graph, this gets converted to two directed edges $u_{out} \rightarrow v_{in}$ and $v_{out} \rightarrow u_{in}$.

Using the Max-Flow Min-Cut theorem, the problem is now equivalent to computing the maximum flow for any two pairs of vertices. We can use the Ford-Fulkerson algorithm \cite{CoLe01} for computing the maximum flow in $O(|E| \cdot \max|F|) = O(|E| |V|)$ as we have that all edge capacities are equal to 1 in the graph $G$. For all pairs of vertices this gives us an algorithm of complexity $O(|E| |V|^3)$.

\medskip

A biconnected graph is a connected and non-separable graph, meaning that if any one vertex were to be removed, the graph will remain connected. The running time can be further speed up by finding all articulation points and biconnected components in time $O(|E| + |V|)$ as only within biconnected components the values of the matrix $P$ can be larger than 1. This can be done in a preprocessing step.

\medskip

We plan to use this efficient algorithm to continue studying the path energy of graphs, in particular bicyclic and biconnected graphs.

\end{document}